%% file: main.tex
\documentclass[a4paper,UKenglish,cleveref, autoref, thm-restate]{lipics-v2021}

\nolinenumbers

\input{macros}

\shortversion{
\title{Brief Announcement: Carry the Tail in Consensus Protocols}
\titlerunning{Brief Announcement: Carry the Tail in Consensus Protocols}}

\fullversion{
\title{Carry the Tail in Consensus Protocols}
\titlerunning{Carry the Tail in Consensus Protocols}}

\author{Suyash Gupta}{University of Oregon \and \url{https://gupta-suyash.github.io/} }{suyash@uoregon.edu}{https://orcid.org/0000-0002-3240-1840}{}

\author{Dakai Kang}{University of California, Davis \and \url{https://dakaikang.github.io/} }{dakang@ucdavis.edu}{https://orcid.org/0000-0002-7867-3681}{}

\author{Dahlia Malkhi}{University of California, Santa Barbara \and \url{https://malkhi.com/} }{dahliamalkhi@ucsb.edu}{https://orcid.org/0000-0002-7038-7250}{}

\author{Mohammad Sadoghi}{University of California, Davis \and \url{https://expolab.org/} }{msadoghi@ucdavis.edu}{https://orcid.org/0000-0003-2779-6080}{This work is partially funded by NSF Award Number 2245373.}

\authorrunning{S. Gupta, D. Kang, D. Malkhi and M. Sadoghi}

\Copyright{Suyash Gupta and Dakai Kang and Dahlia Malkhi and Mohammad Sadoghi}



\shortversion{
\relatedversiondetails[cite={carryfull}]{Full Version}{arxiv.org/xxxx}}

\shortversion{
\category{Brief Announcement}

\EventEditors{Dariusz R. Kowalski}
\EventNoEds{1}
\EventLongTitle{39th International Symposium on Distributed Computing (DISC 2025)}
\EventShortTitle{DISC 2025}
\EventAcronym{DISC}
\EventYear{2025}
\EventDate{October 27--31, 2025}
\EventLocation{Berlin, Germany}
\EventLogo{}
\SeriesVolume{356}
\ArticleNo{54}}

\begin{document}

\maketitle

\begin{abstract}

\begin{wrapfigure}{l}{2cm}
         \includegraphics[width=\linewidth]{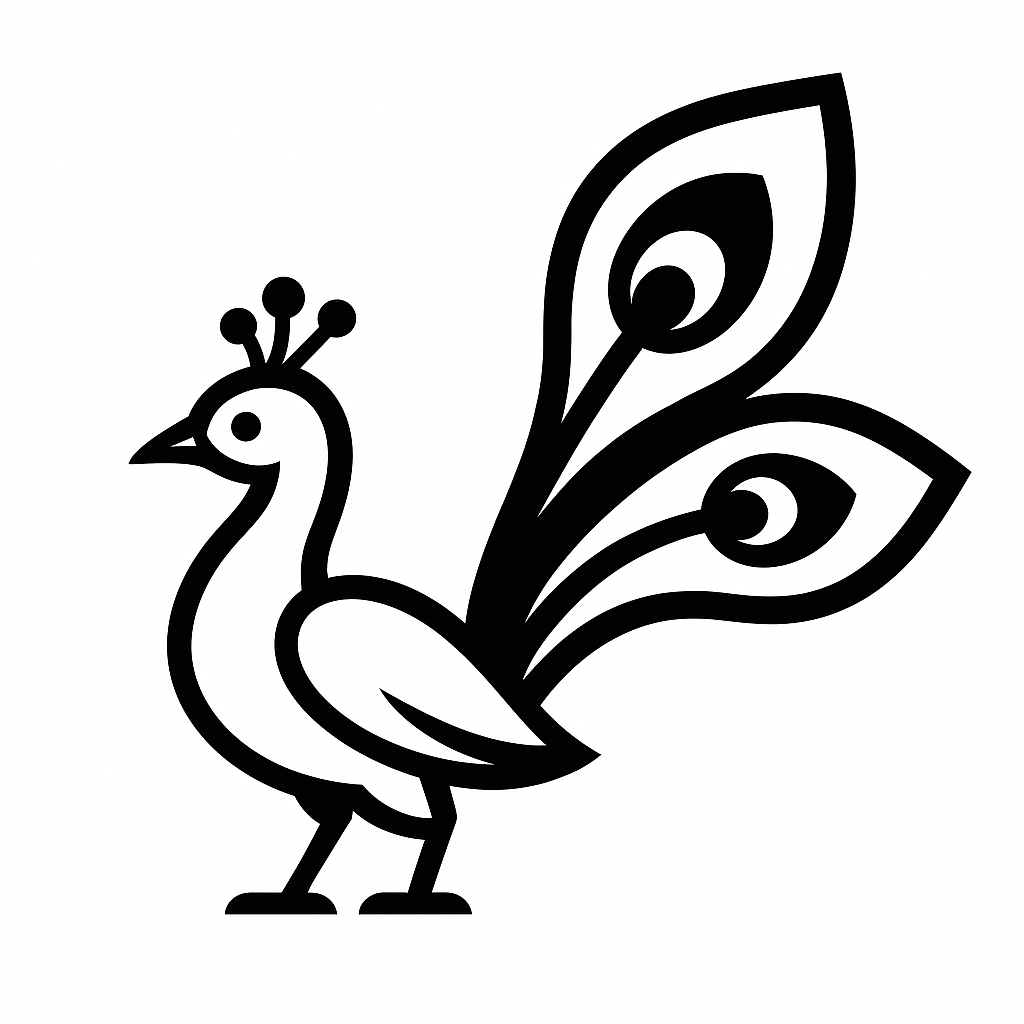}
         \label{fig:peacock}
\end{wrapfigure}   

We present \textbf{\Ctail}, the first deterministic atomic broadcast protocol in partial synchrony that, after GST, guarantees a constant fraction of commits by non-faulty leaders against tail-forking attacks, and maintains optimal, worst case quadratic communication under a cascade of faulty leaders. The solution also guarantees linear amortized communication, i.e., the steady-state is linear.

Prior atomic broadcast solutions achieve quadratic word communication complexity in the worst case. However,
they face a significant degradation in throughput under tail-forking attack. 
Existing solutions to tail-forking attacks require either quadratic communication steps or computationally-prohibitive SNARK generation.  

The key technical contribution is \textbf{\Carry}, a practical drop-in mechanism for streamlined protocols in the HotStuff family. \Carry guarantees good performance against tail-forking and removes most leader-induced stalls, while
retaining linear traffic and protocol simplicity.  
\end{abstract}

\fullversion{
\input{intro}
\input{background}

\input{carry}
\input{related}
\input{conclusion}

}

\shortversion{
\input{ba-intro}
\input{ba-background}
\input{ba-carry}
\input{ba-related}
\input{ba-conclusion}
}

\clearpage
\newpage
\bibliographystyle{plainurl}
\bibliography{DMrefs,xxx}

\newpage
\fullversion{
\appendix
\input{proofs}
}

\end{document}

%% file: macros.tex
\usepackage[utf8]{inputenc}
\usepackage{amsmath, amssymb}
\usepackage{hyperref}
\usepackage[table]{xcolor}
\usepackage{graphicx}
\usepackage[most]{tcolorbox}
\usepackage[ruled,vlined]{algorithm2e}

\usepackage{array, makecell}  
\definecolor{lightgray}{gray}{0.9}

\usepackage{etoolbox}

\usepackage{wrapfig}

\newcommand{\mypp}{\medskip\noindent\textbf}

\newcommand{\inlinecomment}[1]{\hspace{0.5em}{\textcolor{orange}{\texttt{// #1}}}}

\SetKwBlock{Upon}{Upon}{}

\usepackage{etoolbox}
\newcommand{\showme}{yes}

\newcommand{\fullversion}[1]{%
  \ifdefstring{\showme}{yes}{#1}{}%
}
\newcommand{\shortversion}[1]{%
  \ifdefstring{\showme}{no}{#1}{}%
}

\newcommand{\DM}[1]{{}}
\newcommand{\Sadoghi}[1]{}
\newcommand{\Dakai}[1]{}
\newcommand{\Suyash}[1]{}

\newcommand{\Ctail}{Carry-the-Tail\xspace}
\newcommand{\Carry}{Carry\xspace}

\newcommand{\QC}[1]{\ensuremath{\mathsf{QC}}(#1)}
\newcommand{\View}[1]{\ensuremath{{\rm view}}(#1)}

\newtoggle{V1}

\newcommand{\Versioned}[2]{\iftoggle{V1}{#1}{{\color{blue} #2}}}

%% file: intro.tex
\section{Introduction}

Streamlined Byzantine Fault-Tolerant (BFT) consensus protocols, which lie at the core of modern decentralized systems and blockchain applications,
follow HotStuff's~\cite{Yin2019HotStuffBC} pioneering design (``HS-like'') to reach an agreement on the order of executing client transactions among a system of replicas~\cite{hotshot,Chan2023SimplexCA,Chan2020StreamletTS,diembft-hotstuff,Jalalzai2020FastHotStuffAF,Kang2024HotStuff1LC,Malkhi2023HotStuff2OT,Yin2019HotStuffBC}.
This HS-like design is attractive for the following four reasons:
\begin{itemize}
  \item \textbf{Conceptual simplicity:} one block per leader and a chain of quorum certificates (QCs; each certificate formed with the support of a quorum of replicas) guarantees that a transaction has committed.
  \item \textbf{Responsive liveness:} progress after GST without any extra rounds for view-change.
  \item \textbf{Linear communication:} an $O(n)$ word cost in the normal (no failures) case.
  \item \textbf{Censorship protection:} new leader per round prevents clients from being censored.
\end{itemize}

\noindent
Yet, frequent leader rotation makes HS-like protocols fragile whenever a leader is \emph{slow, selfish, or Byzantine}. 
BeeGees~\cite{Giridharan2022BeeGeesSA} shows that in the absence of consecutive honest leaders, bad leaders can cause throughput to collapse through \textit{Tail-Forking} attacking. 
Existing remedial solutions
sacrifice HotStuff's hallmark efficiency by requiring expensive proofs ~\cite{Giridharan2022BeeGeesSA,Jalalzai2025MonadBFTFR,Jalalzai2023VBFTVB}:
an incoming leader needs to collect a quorum of votes for the previous tail (block) or \textbf{prove} that it cannot collect them.
Such a proof needs quadratic communication complexity or requires generating computationally-prohibitive SNARKs.

\mypp{Our Contribution: \Carry}

We introduce \textbf{\Carry}—a \emph{practical}, drop-in mechanism for any HS-like protocol that preserves linear traffic and restores progress under hostile or sluggish leadership.

\smallskip
To explain how the \Carry{} mechanism resolves tail-forking attacking while preserving linear communication complexity, we need to revisit why HS-like protocols require consecutive honest leaders, apart from the block proposer (say $P$), to commit a proposal. 
For example, for HotStuff-2~\cite{Malkhi2023HotStuff2OT}, given a system of $3f{+}1$ replicas, where at most $f$ replicas are Byzantine, two consecutive honest leaders allow HotStuff-2 to form a chain of two QCs, which guarantees the following:
\begin{enumerate}
  \item the first QC (\emph{lock}), guarantees that at least $2f{+}1$ replicas {\em voted} for the proposed block,
  \item the second QC (\emph{lock-commit}), guarantees that at least $2f{+}1$ replicas \textit{voted} for the block that extends the lock (consecutively succeeds the block by $P$).
\end{enumerate}

In HS-like protocols, a replica's role is to {\em guard its lock}: replicas send their current lock to an incoming leader through \texttt{NEW-VIEW} messages and require future proposals to extend their lock (or higher). 
Replicas also send their votes to the next leader in the \texttt{NEW-VIEW} messages (or an "empty" vote if they time-out). However, they {\em do not guard their vote} during this leadership change (view-change). 
Consequently, if the leader responsible for aggregating the votes for the first QC is Byzantine, 
it can prevent the lock by proposing a succeeding block that does not extend the block by $P$.
In essence, this is referred to as the tail-forking attack. \DM{The notation is confusing misleading, because above, we use $P$ to denote the proposal by $P$ which has a lock, and only the second phase is being forked. We need to phrase the paragraph above better.}

\smallskip
The \Carry mechanism protects against tail-forking by treating replica votes as first-class citizens: 
the votes not only help to reach agreement, but also transport knowledge of the last {\em safe block} (proposed by an honest leader) forward.
This strategy has a dual benefit: On the one hand,
a Byzantine leader cannot skip a safe block because the \Carry mechanism forces a leader that skips a safe block to generate a proof that it did not receive any information regarding the safe block.
On the other hand, 
an honest future leader can \textit{reinstate} a previous safe block even if it does not have $2f{+}1$ votes for it.

It is worth noting that this mechanism aligns with the BeeGees~\cite{Giridharan2022BeeGeesSA} philosophy: use the votes not merely to tally agreement, but to transport knowledge of the last safe block forward. However, \Carry embodies a distinct implementation that preserves the linear communication complexity characteristic of HS-like protocols. In essence, a \Carry leader need not collect or prove the absence of $2f{+}1$ votes to propagate the previous tail. In prior approaches, the obligation for a leader to demonstrate such a proof was the root cause of quadratic overhead. 

\mypp{The \Ctail Solution.}

We present a full solution for atomic broadcast called \Ctail that incorporates \Carry into HotStuff-2. Briefly, \Ctail operates a view-by-view regime. At the beginning of each view, replicas send the incoming leader their votes for each of the most recent $\rho$ views\footnote{One can optimize and send votes for fewer views in most cases, e.g., when a proposal extends an immediate predecessor; we omit this for brevity.}.
If it abstained, a replica issues a signed empty share. The total payload across $n$ replicas is $O(\rho) \times n$. 

The leader sends replicas a proposal extending the highest pending previous proposal. It uses the vote-information it collected to justify (only) the views it skips, thus enabling a linear solution. More specifically, there are two cases:

\begin{enumerate}
    \item \textbf{Easy case:} If the leader has collected $2f{+}1$ votes for the highest previous proposal, it bundles them into a QC and extends that QC in its proposal—no further evidence is needed.

    \item \textbf{Hard case:} If there is no fresh QC, existing HS-like protocols fall back to quadratic communication to prove the absence of higher certificates. \Carry avoids this blow-up by reinstating forward the single highest vote that extends the leader’s highest known QC: \begin{itemize}
        \item The leader reinstates the full block corresponding to that highest vote.
        \item To guarantee that there could be no higher vote, for each view between the highest vote and the current view, it also aggregates $2f{+}1$ empty vote shares into an empty certificate, proving that no quorum could have formed there.
        \item All these attachments remain bounded by $O(n \cdot \rho)$.
    \end{itemize}
\end{enumerate}

A replica accepts a leader proposal as justified if it incorporates the replica's highest vote. Accordingly, when it starts the next view, the replica sends to the next incoming leader a \texttt{NEW-VIEW} message containing its vote for this view or an empty-share.


\mypp{\Ctail Properties.} 
%

The key property \Ctail achieves is that only a consecutive succession of more than $\rho$ bad leaders could prevent a proposal amidst them from being reinstated or extended by the next successful view proposal:


\begin{flushleft}
\textbf{Definition 1} (\emph{$\rho$-tail-resilience}). \emph{We say that a proposal $T$ is \textit{$\rho$-isolated} if view of $T$ is situated among a succession of $\rho$ consecutive bad leaders.
(After GST,) each proposal $T$ by an honest leader, which is not $\rho$-isolated, is guaranteed to be included in the global sequence.}
\end{flushleft}

%

In other words, if there is no unlucky succession of more than $\rho$ views preceding a tail, the tail will not be forked. By the pigeon-hole principle, a rotation of $3f{+}1$ leaders has $2f{+}1 - (f/\rho)$ leaders that obtain $\rho$-tail-resilience against a worst-case scenario of $f$ bad leaders interspersed adversarially. For a reasonably small choice of (say) $\rho=6$, only a fraction $\frac{1}{12}$ of honest leader proposals, in theory, get forked. It is worth noting that a consecutive succession of $6$ bad leaders is highly unlikely in practice, especially if the leader rotation is randomized.

$\rho$-tail-resilience is related to a property in BeeGess~\cite{Giridharan2022BeeGeesSA} called \textit{``Any-Honest-Leader commit''} (AHL). Under AHL, (after GST) every block by an honest leader eventually gets committed. On its own, $\rho$-tail-resilience is weaker than AHL, but it suffices for censorship resistance and effective good-put while preserving linearity. Importantly, whereas BeeGees employs a complex leader hand-off to satisfy AHL, \Ctail is advantageous in practice.
Below we discuss an additional property that \Ctail maintains and BeeGees does not.

As for liveness, in \Ctail, a consecutive pair of honest leaders, a two-chain, ensures a commit decision. 
Consequently, \Ctail guarantees an unbounded number of commit blocks and progress:


\begin{flushleft}
\textbf{Definition 2} (\emph{Liveness}). After GST, there is an unbounded number of committed blocks proposed by honest leaders.\footnote{despite potential tail-forking attacks}
\end{flushleft}

\Ctail achieves liveness and tail-resilience while maintaining a key tenet of the HS-like protocol family: the communication incurred by the leader handover protocol is bounded by $O(\rho \cdot n)$. 


\begin{flushleft}
\textbf{Definition 3} (\emph{Linearity}). \emph{After GST, a commit decision incurs $O(f_a \cdot n)$ word-communication cost in face of $f_a$ actual faults, and every sequence of $O(n)$ commit decisions incurs $O(n^2)$ communication.}
\end{flushleft}

\mypp{Discussion.} 
The \Carry mechanism tackles an additional problems that revolves around leader \textit{slowness}: \textbf{stragglers}, which arises when leaders are slow to propose due to benign environmental issues. Such stragglers may slow protocol progress for everyone and, moreover, they might miss proposing their own blocks altogether. Consensus systems typically allow for a generous time slot per leader to prevent expiration on stragglers. Consequently, when a real fault occurs, they are slow to react. 

By enforcing that every new leader reinstates its highest vote and provides explicit proofs of emptiness for intervening views, \Carry eliminates leader-induced stalls and aids slow or straggling replicas. Importantly, a proposal from a slow leader can be reinstated and eventually committed even if it does not receive $2f{+}1$ votes in its original view. 

In summary, \Carry restores HotStuff’s trademark efficiency while closing its leader-performance gap.


%% file: background.tex
\section{Background}
This paper focuses on protocols that are designed for solving Byzantine Fault Tolerant atomic broadcast (aka BFT consensus), in the standard partial synchrony model with $n=3f{+}1$ replicas, where at most $f$ replicas are Byzantine (e.g., see~\cite{Kang2024HotStuff1LC,Malkhi2023HotStuff2OT,Yin2019HotStuffBC}).

\subsubsection*{Streamlined, Block-based Protocols: Key concepts}

We start with a brief recollection of two key concepts in HS-like protocols: \textit{Streamlined} and \textit{Block-Based}.  

A streamlined block-based consensus protocol is a type of BFT consensus algorithm that simplifies and accelerates block production in distributed systems such as blockchains. 
Traditional BFT protocols, such as PBFT~\cite{pbftj}, use multiple distinct phases (e.g., pre-prepare, prepare, commit) to reach agreement and often require complex view changes and high communication overhead (quadratic in some cases).
In contrast, streamlined protocols:
\begin{itemize}
    \item Reduce the number of phases, often reusing or overlapping them.
    \item Embed view changes naturally into the block proposal process, avoiding separate and expensive view-change protocols.
    \item Typically rotate leaders frequently (every view), even in the absence of faults.
\end{itemize}
This streamlined design improves throughput (overlapping phases), latency (faster finality), simplicity (fewer protocol states and messages), and efficiency (linear communication complexity).
Streamlined protocols are typically coupled with a Block-Based design. In Block-Based protocols:
\begin{itemize}
    \item Each proposal is a block that contains transactions and metadata (e.g., QCs).
    \item The blockchain is a chain of quorum-certified blocks, each building on the previous one.
    \item The safety of decisions (i.e., agreement on committed blocks) is ensured by how these blocks and certificates reference one another.
\end{itemize}

HotStuff, the first streamlined and block-based protocol, incurs only linear communication cost per leader block and avoids the expensive view-change regime of PBFT.
Its streamlined variant proposes a new block every view, rotating leaders every phase and avoiding extra recovery rounds.

\mypp{Blocks and Metadata.}
A HS-like protocol works through views.
In each view $v$, one replica is designated as the leader and is responsible for proposing a block.
We denote the leader of view $v$ by $L_v$ and its proposed block by $B_v$; conversely, we denote the view $v$ of a block by $\View{B_v}$.


A key ingredient in the protocol is (a leader) forming a \textit{quorum-certified} block:
\begin{itemize}
    \item A \textit{Quorum Certificate (QC)} is a cryptographic proof that a block is approved by a quorum of replicas. It is formed by aggregating $2f{+}1$ valid votes (signature shares\footnote{To generate and aggregate signature shares, HS-like protocols make use of threshold signatures schemes.}) from replicas in a system of $n \geq 3f{+}1$.
    \item QCs ensure safety by proving that a quorum of replicas agreed on a proposal. We will interchangeably denote the QC for a proposal from view $v$ by $\QC{B_v}$ or $\QC{v}$. 
\end{itemize}

Leaders form a QC in HotStuff-2 with linear communication overhead by borrowing a technique from~\cite{Cachin2001SecureAE,Reiter1994SecureAP}
: a sender that (i) disseminates a block proposal, (ii) collects certified acknowledgments (votes) that are signed with signature-shares, and (iii) aggregates these shares via threshold cryptography. The sender can disseminate the QC to all replicas at a linear word-communication cost. 

Blocks are chained to one another through QCs and other metadata. This allows replicas to detect conflicts (e.g., forks) and enforce voting rules to maintain consistency: 

\begin{itemize}
    \item Each block $B$ extends a previously certified block $B'$ by including the certificate $\QC{B'}$ as $B.qc$.
    \item Unique to \Carry, a block $B$ may \textit{reinstate} an uncertified previous block $B'$,  without having a QC for it, by embedding it (in full) as a \textit{reinstated block} in $B.reinstate$. By slight abuse of terminology, we also say that $B$ extends $B'$.
    \item The transitive closure of the two ``extend'' relationships is denoted '$\succeq$'. 
\end{itemize}

\subsubsection*{Streamlined HotStuff-2 in a Nutshell}
\begin{figure}
\centering
\includegraphics[width=.7\linewidth]{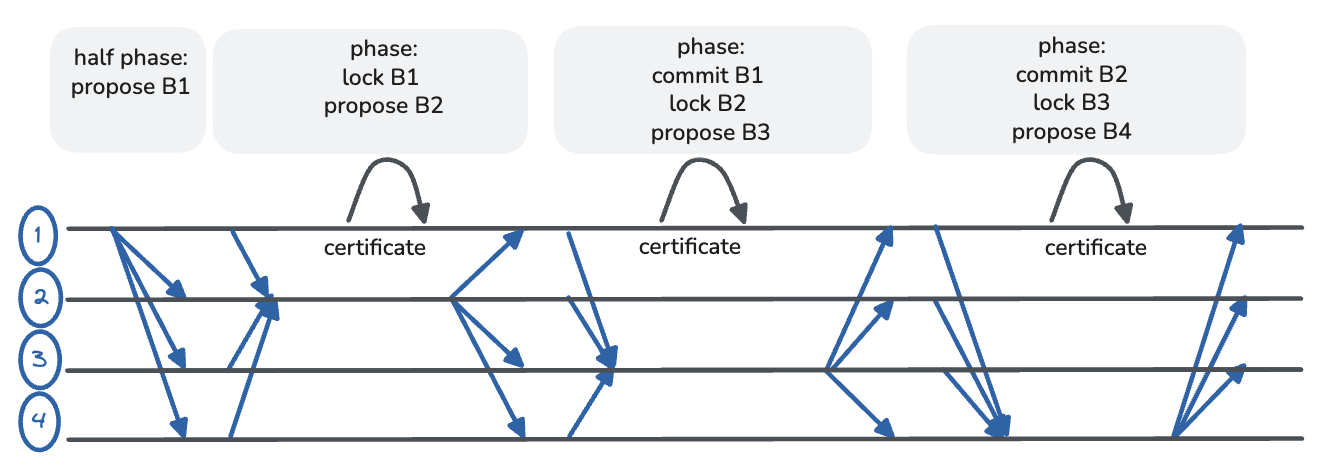}
\caption{Streamlined HotStuff-2 protocol flow.} \label{fig:shs2}
\end{figure}   

We use streamlined HotStuff-2~\cite{Malkhi2023HotStuff2OT}, a variant of HotStuff that reduces its latency by two-half phases, to exemplify the \Carry mechanism. The work flow of HotStuff-2 is depicted in Figure~\ref{fig:shs2}. 
{\em Note:} \Carry can be applied to other streamlined and block-based protocols.

The key safety rules of the protocol are defined below.

\mypp{Leader Proposal Rule.}
A leader for view $v$ may propose a block once any of the following conditions are met:
\begin{enumerate}
\item It forms a fresh quorum certificate $\QC{w}$ by aggregating $2f{+}1$ \texttt{NEW-VIEW} messages with identical votes for some prior view $w < v$.
\item It receives $3f{+}1$ \texttt{NEW-VIEW} messages or the Pacemaker indicates that sufficient time has passed to collect messages from all honest replicas.
\end{enumerate}

These conditions ensure safety by extending the highest known QC, and liveness by allowing the leader to eventually propose even under partial synchrony.

\mypp{Voting Rule.}
A replica follows a Voting Rule to decide whether to accept and vote for a leader's proposal. This decision ensures that all honest replicas remain consistent; they only vote for proposals that do not conflict with previously locked blocks.

Specifically, each replica maintains a \textit{lock}, which is the highest QC attached to a block it has previously voted for. A replica votes for a proposed block $B$ if and only if $B.qc$ has a higher or equal view than $\text{lock}$. Then it locks on $B.qc$. 

\mypp{Commit Rule.}
A Commit Rule in HotStuff-2 requires a chain of two consecutive QCs to commit a proposal. 
The replica commits a block $B_v$ to the total order if it receives $B_{v+1}$ and $B_{v+2}$, with $B_{v+2}.qc = B_{v+1}$ and $B_{v+1}.qc = B_v$. 

\toggletrue{V1}
\input{hs2-alg}

\mypp{Protocol Flow.}
Algorithm~\ref{alg:hs2} provides a pseudo-code description of HotStuff-2. The general flow of the protocol is view by view, as depicted in Figure~\ref{fig:shs2}, is as follows: 
\\
\begin{description}

\item[\textbf{Replica $\longrightarrow$ incoming leader.}]
At the beginning of a new view, each replica sends the incoming leader a \texttt{NEW‑VIEW} message that includes:
\begin{enumerate}
    \item the next view number,
    \item the replica’s highest QC (its \emph{lock}), and
    \item its vote as per the Voting Rule. 
\end{enumerate}

\item[\textbf{Leader $\longrightarrow$ replicas.}]
The leader sends replicas a proposal according to the Leader Proposal Rule, that consists of a block $B_v$ that extends the highest QC it aggregated.

\item[\textbf{Voting.}]
Next, replicas decide whether to accept the proposal based on the Voting Rule. 
At the end of each view, the replicas proceed to perform the handover protocol with the next incoming leader.

\end{description}

The pseudo-code in Algorithm~\ref{alg:hs2} makes use of a \texttt{Pacemaker} API: 
a replica invokes \texttt{Pacemaker.advanceView}$()$ when it wishes to exit the current view; \texttt{Pacemaker.exitView}$()$ notifies a replica that the current view has expired; \texttt{Pacemaker.syncView}$()$ notifies the replica that under synchrony conditions, all other replicas have entered the next view and their \texttt{NEW-VIEW} messages have been delivered.

\subsection*{The Leader-Slowness Problem and the Tail-Forking Attack}

Two principal problems—\textit{leader-slowness} and \textit{tail-forking}—undermine the performance benefits of streamlined protocols.

\begin{itemize}
    \item {\em Leader-slowness} hurts liveness and latency under unfavorable network conditions or due to hardware capacity limitation (see Figure~\ref{fig:tailfork}). 
   This issue arises when leaders experience delays in proposing blocks due to benign environmental factors. Such slow leaders can impede the overall progress of the protocol and may even miss the opportunity to propose their own blocks entirely. To accommodate these delays, consensus protocols often allocate generous time slots to each leader, ensuring that slow leaders are not prematurely timed out. However, this leniency also results in sluggish responsiveness when actual faults occur.

    \item {\em Tail-forking} is a performance/liveness vulnerability in streamlined HS-like protocols (depicted in Figure~\ref{fig:tailfork}).
In these protocols, the commitment of a proposal depends on the actions of the next two consecutive leaders.
A malicious leader can disrupt the commit phase by refusing to acknowledge or extend the proposal of the preceding (honest) leader.
This attack can cause forking at the tail of the chain, preventing the system from committing otherwise valid proposals.
It does not necessarily cause complete censorship, but it leads to increased latency for certain clients and reduced throughput.

\end{itemize}

\begin{figure}[t]
    \centering
    \includegraphics[width=.65\linewidth]{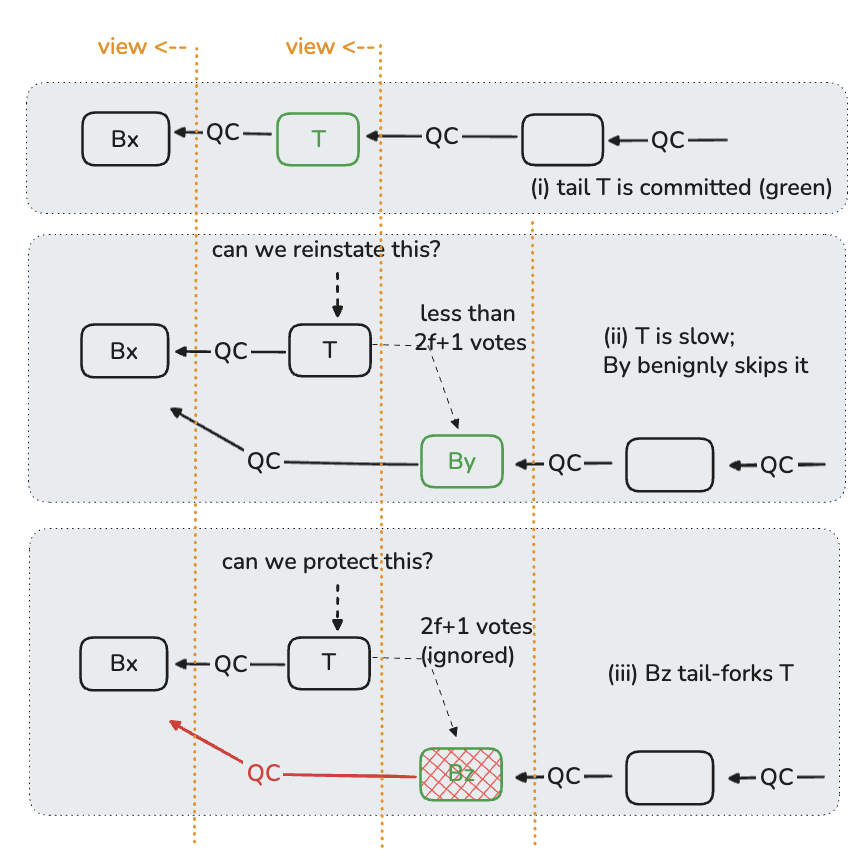}
    \caption{HotStuff-2 scenarios: normal (top), leader slowness (middle) and tail-forking (bottom).}
    \label{fig:tailfork}
\end{figure}

These issues are the core motivations for this work. The \Carry mechanism that we introduce mitigates these issues while preserving the streamlined, efficient design of HS-like protocols.



%% file: hs2-alg.tex
\begin{algorithm}[t]
\DontPrintSemicolon
\SetAlgoVlined
\KwData{Each replica maintains its 
$highest{\text -}vote$ and a \textit{lock} on the QC attached to its highest{\text -}vote
}

\Begin(view $v$){
\uIf{replica is $L_v$ \inlinecomment{Leader logic}}{ 
  wait until LeaderProposalConditions($v$) is true \;
  create a block $B_v$, AttachMetadata($B_v$) and propose $B_v$ \;
}

\BlankLine
\Upon(receiving proposal block $B_v$ \inlinecomment{Replica logic (including leader)}){
    \uIf{VotingRule($B_v$)}{$lock \leftarrow view(B_v.qc)$\; \Versioned{}{$highest{\text -}vote \leftarrow B_v$}}
  Pacemaker.advanceView() \;
  CommitRule($B_v$) \;
}

\BlankLine
\Upon(Pacemaker.exitView\mbox{()} notification \inlinecomment{View-change logic}){
advance current view number to $v$\;
send \texttt{NEW-VIEW} message to next leader with: \;
\Indp 
    (i) current view number \;
    (ii) $lock$ \;
    (iii) a signature-share on $highest{\text -}vote$ \;
    \Versioned{}{(iv) up to $\rho$ signature-shares on empty($w$) for each view $w$ between $highest{\text -}vote$ and current view\;}
\Indm
}
%
} 
\BlankLine
Define \textbf{VotingRule($B_v$)} as true if: \;
\Indp 
    $view(lock) \le view(Bv.qc)$  \Versioned{}{and $CarryRule(B_v)$}
\Indm

\Versioned{}{
\BlankLine
Define \textbf{CarryRule($B_v$)} as true if: \; 
\Indp 
$v > \View{B_v.qc} + \rho$ \inlinecomment{vacuously true, not eligible for reinstating} \textbf{OR} \;
for $T$ denoting $highest{\text -}vote$: ($B_v$ extends $T$ AND $T \succeq lock$) AND \;
$B_v$ includes empty-certificates for each view between $\max\{\View{B_v.reinstated}, \View{B_v.qc} \}$ and $v$ \;
\Indm
}

\BlankLine
Define \textbf{CommitRule($B_v$)}: \; 
\Indp 
\uIf{$B_v$ has an attached QC for $B_{v-1}$, and $B_{v-1}$ has an attached QC for $B_{v-2}$}{
  commit the uncommitted prefix of the chain up to (incl.) $B_{v-2}$\; 
}
\Indm

\BlankLine
Define \textbf{LeaderProposalConditions($v$)} as true if any of the following hold: \; 
\Indp 
(i) a fresh QC is formed from $2f{+}1$ \texttt{NEW-VIEW} messages with identical highest votes for some $w < v$\; 
(ii) $3f{+}1$ \texttt{NEW-VIEW} messages are received or receive \texttt{Pacemaker.syncView}$()$ notification\;
\Indm

\BlankLine
Define \textbf{AttachMetadata($B_v$)}: \; 
\Indp 
$B_v.qc \leftarrow$ highest known QC\;
\Versioned{}{
\uIf{$\View{B_v.qc} + \rho \leq v$ \inlinecomment{eligible for reinstating}}{
  $T \gets$ highest known vote extending $B_v.qc$\;
  \uIf{$T \succ B_v.qc$}{$B_v.reinstated \leftarrow T$ \;}
  attach to $B_v$ empty-certificates for each view between $\View{T}$ and $v$ \;
  }
\Indm
}

\Versioned{
\caption{(Streamlined) HotStuff-2 Protocol} \label{alg:hs2}
}{
\caption{\Ctail Protocol (differences from HotStuff-2 in blue)} \label{alg:carry}
}
\end{algorithm}

%% file: carry.tex
\section{Carry}
\label{sec:carry}

The \Carry mechanism improves view synchronization and eliminates most
leader‑induced stalls.  The core idea is straightforward: the next leader must
\emph{reinstate} the last uncertified block it sees, so that progress
is not impinged by the second QC's diffusion.

Recall that the core mechanism in HS-like protocols is a lock-commit regime: $2f{+}1$ replicas need to hold a lock on block $B_x$, proposed in view $x$, in order to commit it in the future. 
A replica acquires its lock on $B_x$ precisely when it votes for a tail block $T$ that extends $B_x$. For example, in HotStuff-2, which we use here to exemplify the \Carry mechanism, $T$ would directly extend $B_x$, i.e., $T.qc = \QC{B_x}$. 
Although those votes are shipped in \texttt{NEW‑VIEW} messages, nothing
in HotStuff-2 requires the next leader to honor them; a sluggish leader may fail to collect $2f{+}1$ votes, and a Byzantine leader may ignore them,
causing a \emph{tail‑forking} attack.

\mypp{Core insight.}

Votes are first-class citizens of the consensus process; replicas must guard their vote. 
Each vote implicitly endorses both the certified block and the extending block.  \Carry makes that endorsement \emph{explicit}: replicas ship their latest
vote (e.g., $T$) into the next view. The incoming leader proposal $B_v$ must \textit{reinstate} $T$, if needed, as a {\em reinstated block} $B_v.reinstated = T$.
Thus, both the certified block \emph{and} its immediate extension are
protected; $2f{+}1$ honest votes cannot be silently dropped.

\subsection*{\Carry Implementation}

The key safety rules in \Carry that are different from HotStuff-2 are defined below.

\togglefalse{V1}
\input{carry-alg}

\mypp{Carry Rule.}
A valid proposal $B_v$ from the leader $L_v$ of view $v$ has the following format:
\begin{description}
    \item[Highest QC:] the highest quorum certificate $\QC{x}$ known to the leader $L_v$ is attached as $B_v.qc = \QC{x}$; 
    \item[Reinstate:] if $L_v$ received fewer than $2f{+}1$ votes for the highest tail $T$ extending $B_x$, and $v - x \leq \rho$, then $T$ is reinstated (in full) as $B_v.reinstated$;
    \item[Justification:] If $v - x \leq \rho$, then empty-certificates are attached to $B_v$ for each view, which $B_v$ does \textbf{not} extend, between views $x$ and $v$. It is worth noting that if $B_v.reinstated$ exists, empty-certificates for views preceding the reinstated block are recursively attached within the chain of reinstated blocks.
\end{description}

\mypp{Voting-Rule}
       
A replica accepts the proposal $B_v$ of leader $L_v$ of view $v$ if:
\begin{enumerate}
    \item $B_v.qc$ has a higher or equal view than the replica's $lock$,
    \item $B_v$ adheres to the Carry Rule above.
\end{enumerate}

\Carry retains from HotStuff-2 the Commit Rule and Leader Proposal Rule.

\subsection{The \Ctail Protocol}

Akin to HotStuff-2, the \Ctail protocol flows view-by-view, as detailed in Algorithm~\ref{alg:carry}.
At the end of each view, the replicas and the incoming leader perform a handover protocol as follows:

\begin{description}

\item[\textbf{Replica $\longrightarrow$ incoming leader.}]

Each replica sends an incoming leader a \texttt{NEW‑VIEW} message that carries
\begin{enumerate}
    \item the next view number
    \item the replica’s highest QC (its \emph{lock})
    \item its vote-shares (possibly empty) in the past $\rho$ views.\footnote{It is possible to send less information if the lock precedes the current view by less than $\rho$ views; we omit this optimization for simplicity.}
\end{enumerate}
      
\item[\textbf{Leader $\longrightarrow$ replicas.}]

On satisfying the Leader Proposal Rule, the leader of view $v$ proposes a block $B_v$ that   
\begin{enumerate}
    \item extends the highest QC it has collected, potentially freshly aggregated from \texttt{NEW-VIEW} messages,  as $B_v.qc$,
    \item reinstates $T$ in full as $B_v.reinstated$, provided that the highest QC is from the past $\rho$ views and the highest voted block $T$ extends the highest QC, 
    \item attaches empty-certificates for each view between $\View{T}$ and $v$.

\end{enumerate}

\end{description}

Note that if a leader receives two conflicting highest votes, it reinstates
\emph{both}; honest replicas will recognize the view as
faulty and drop the conflicting votes. We omit the details from Algorithm~\ref{alg:carry} for brevity.

\subsection*{Correctness Proofs}

We envision \Carry as a drop-in mechanism for any HS-like protocol. 
In this paper, we use HotStuff-2 to explain the design of \Carry. 
Thus, \Ctail inherits the safety and liveness properties of HotStuff-2, since the \Carry mechanism only boosts performance but does not change the basic safety and liveness rules. We defer these proofs to Appendix~\ref{s:appendix} and concentrate here on proving the key novel property of \Ctail, $\rho$-tail-resilience.

\begin{theorem}
    After GST, a tail block $T$ proposed by an honest leader that is not $\rho$-isolated and receives $2f{+}1$ votes will be extended by the next honest view. 
\end{theorem}

\begin{proof}
    As $T$ is not $\rho$-isolated, there exists two honest views $x$ and $y$, such that $x < \View{T} < y \leq x + \rho$. We want to show that the honest leader $L_y$'s proposal $B_y$ extends $T$ in the chain, $B_y$.qc $= T$.  
    In particular, if \(B_y\) receives votes from all honest replicas before the view expires, then every honest replica locks on \(B_y.qc\), and \(T\) will commit as soon as a later block in the chain is committed.

    {\em Easy Case:} When $y = \View{T}{+}1$, $L_y$ forms a QC for $T$ as the highest QC, proposes a block $B_y$ that extends $\QC{T}$ and disseminates it. All honest replicas become locked on $B_y.qc$, and $T$ is guarded from ever being skipped. 

    {\em Other Case:} When $y > \View{T}+1$ and $L_y$ obtains a QC for a view higher than $\View{T}$. 
    We need to show that $B_y$ does not skip $T$, i.e., that $B_y.qc \succ T$. 
    We do so by showing that between views $x$ and $y$ if a proposal $P$ extends $T.qc$ and receives $f{+}1$ honest votes, then $P$ is extended by the next valid proposal $P'$. {\em Note:} since $L_x$ is honest and $T.qc$ is from a view $x$ or higher, the number of views between $T.qc$ and $P'$ do not exceed $\rho$.

    By the Carry Rule, if $P'.reinstated$ exists, $P'$ must include an empty-certificate as a justification for each view between $P'.reinstated$ and $P'$. Because $P'$ is the next valid proposal succeeding $P$, $P'.reinstated$ has to be from a view not higher than $P$. But since $P$ has $f{+}1$ honest votes, it is impossible for $P'$ to have an empty certificate for it. Hence, $P'.reinstated$ must be $P$. 
If $P'.reinstated$ does exist, then $P'$ must have a fresh QC. Again, by minimality, $P'.qc$ must not be higher than $P$, but cannot skip $P$ either. Hence, in this case, $P'.qc$ is for $P$.

From the argument above, there may be a succession of valid proposals between $x$ and $y$ extending $T.qc$ that are chained to each other via QCs or reinstated blocks. $T$ is within this chain, and $L_y$ terminates it. We obtain that $L_y \succ $T.
\end{proof}

\begin{figure}[t]
    \centering
    \includegraphics[width=0.75\linewidth]{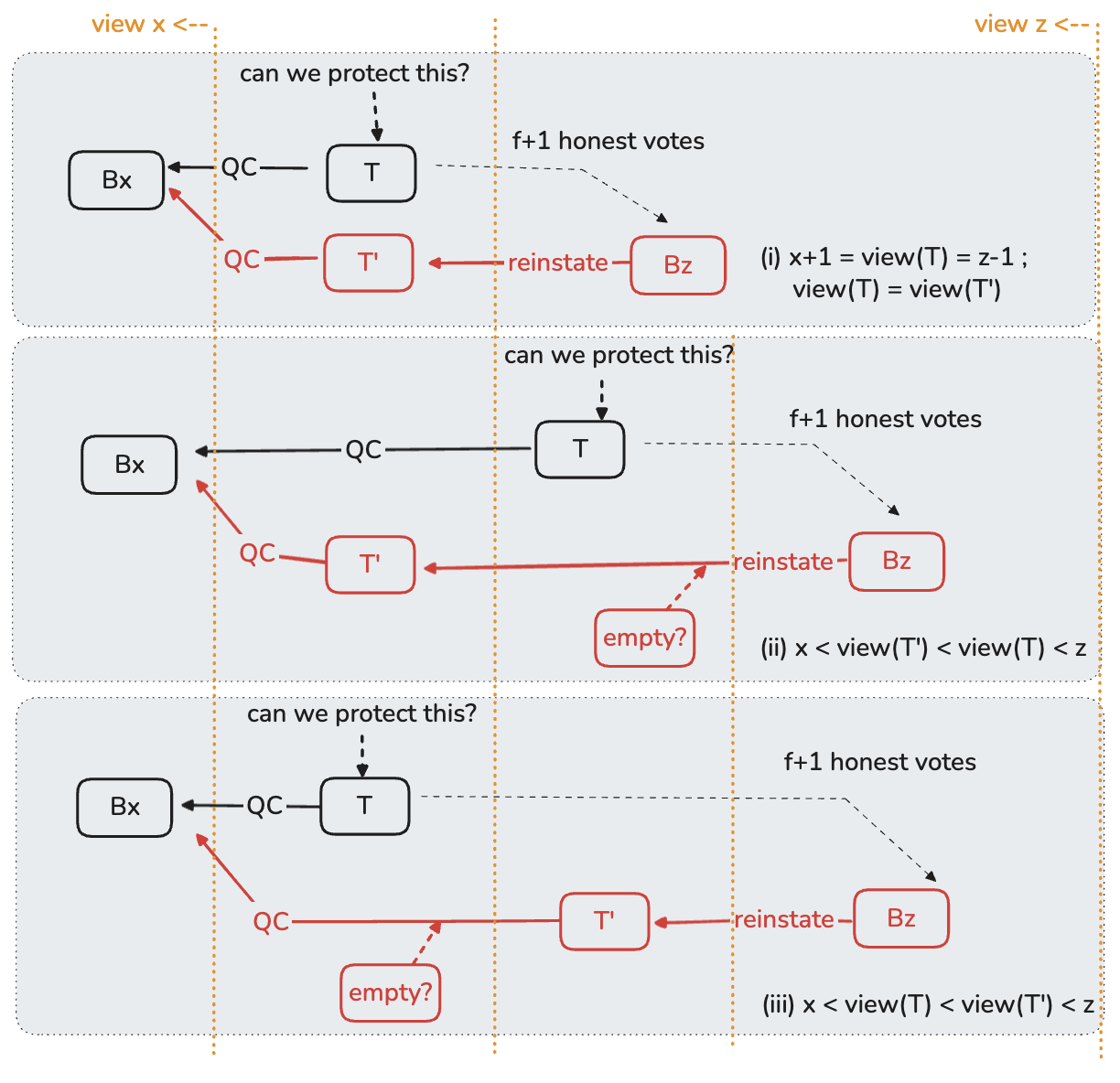}
    \caption{Attempts by $L_z$ to hide $T$ by attaching an invalid reinstated-block $T'$ which skips $T$.}
    \label{fig:carry}
\end{figure}

\mypp{Illustration.}

Consider the case of an honest tail $T$ referencing a QC from a view $x < \View{T}$. We want to prevent a bad leader in a future view $z$ from tail-forking $T$. 
We assume that views $x$ and $z$ are not more than $\rho$ views apart, i.e., $z \leq x + \rho$.
Figure~\ref{fig:carry} illustrates diagrammatically 
how the \Carry mechanism prevents tail-forking of $T$ under three possible cases:

\begin{description}
  \item[(i) Consecutive views.]
        If $z = x + 2$, then $\mathrm{view}(T') = \mathrm{view}(T)$ and forking is impossible.
  \item[(ii) Skipping forward.]
        If $\mathrm{view}(T') < \mathrm{view}(T)$, the $f{+}1$ honest replicas who voted for $T$
        will refuse to vote for $B_z$ due to the lack of empty-certificate for $\mathrm{view}(T)$.
  \item[(iii) Skipping backward.]
        If $\mathrm{view}(T') > \mathrm{view}(T)$, those same $f{+}1$ honest replicas will not sign an
        empty-certificate for view $\mathrm{view}(T)$; $L_z$ thus lacks the proof
        mandated in the Carry Rule: it will miss an empty-certificate for $\mathrm{view}(T)$ which is between the locked-view $x$ and the alleged highest reinstated $\mathrm{view}(T')$. 
\end{description}

Hence, any successful fork needs \emph{$\rho{+}1$ consecutive Byzantine leaders}.
For all but this extreme scenario, \Carry forces each proposal to acknowledge the highest vote or certificate held by honest replicas. This prevents faulty leaders from making unilateral progress while bypassing honest blocks, thereby neutralizing tail-forking attacks under synchrony.

In addition, Carry mitigates the Leader-Slowness issue where a benign straggler leader $L_v$ fails to send $B_v$ to $2f+1$ honest replicas before timeout. It is guaranteed that $B_v$ is not skipped as long as at least $f+1$ honest replicas vote for $B_v$ and $L_v$ is not \emph{$\rho$-isolated}.


%% file: carry-alg.tex
\begin{algorithm}
\DontPrintSemicolon
\SetAlgoVlined
\KwData{Each replica maintains its 
$highest{\text -}vote$ and a \textit{lock} on the QC attached to its highest{\text -}vote
}

\Begin(view $v$){
\uIf{replica is $L_v$ \inlinecomment{Leader logic}}{ 
  wait until LeaderProposalConditions($v$) is true \;
  create a block $B_v$, AttachMetadata($B_v$) and propose $B_v$ \;
}

\BlankLine
\Upon(receiving proposal block $B_v$ \inlinecomment{Replica logic (including leader)}){
    \uIf{VotingRule($B_v$)}{$lock \leftarrow view(B_v.qc)$\; \Versioned{\DM{should assign highest-vote here, too}}{$highest{\text -}vote \leftarrow B_v$}}
  Pacemaker.advanceView() \;
  CommitRule($B_v$) \;
}

\BlankLine
\Upon(Pacemaker.exitView\mbox{()} notification \inlinecomment{View-change logic}){
advance current view number to $v$\;
send \texttt{NEW-VIEW} message to next leader with: \;
\Indp 
    (i) current view number \;
    (ii) $lock$ \;
    (iii) a signature-share on $highest{\text -}vote$ \;
    \Versioned{}{(iv) up to $\rho$ signature-shares on empty($w$) for each view $w$ between $highest{\text -}vote$ and current view \DM{I think we should send $\rho$ empties, regardless of highest-vote} \;}
\Indm
}
%
} 
\BlankLine
Define \textbf{VotingRule($B_v$)} as true if: \;
\Indp 
    $view(lock) \le view(Bv.qc)$  \Versioned{}{and $CarryRule(B_v)$}
\Indm

\Versioned{}{
\BlankLine
Define \textbf{CarryRule($B_v$)} as true if: \; 
\Indp 
$v > \View{B_v.qc} + \rho$ \inlinecomment{vacuously true, not eligible for reinstating} \textbf{OR} \;
for $T$ denoting $highest{\text -}vote$: ($B_v$ extends $T$ AND $T \succeq lock$) \DM{need to address the case that $B_v$ extends a QC highers than $lock$} AND \;
$B_v$ includes empty-certificates for each view between $\max\{\View{B_v.reinstated}, \View{B_v.qc} \}$ and $v$ \;
\Indm
}

\BlankLine
Define \textbf{CommitRule($B_v$)}: \; 
\Indp 
\uIf{$B_v$ has an attached QC for $B_{v-1}$, and $B_{v-1}$ has an attached QC for $B_{v-2}$}{
  commit the uncommitted prefix of the chain up to (incl.) $B_{v-2}$\; 
}
\Indm

\BlankLine
Define \textbf{LeaderProposalConditions($v$)} as true if any of the following hold: \; 
\Indp 
(i) a fresh QC is formed from $2f{+}1$ \texttt{NEW-VIEW} messages with identical highest votes for some $w < v$\; 
(ii) $3f{+}1$ \texttt{NEW-VIEW} messages are received or receive \texttt{Pacemaker.syncView}$()$ notification\;
\Indm

\BlankLine
Define \textbf{AttachMetadata($B_v$)}: \; 
\Indp 
$B_v.qc \leftarrow$ highest known QC\;
\Versioned{}{
\uIf{$\View{B_v.qc} + \rho \leq v$ \inlinecomment{eligible for reinstating}}{
  $T \gets$ highest known vote extending $B_v.qc$\;
  \uIf{$T \succ B_v.qc$}{$B_v.reinstated \leftarrow T$ \;}
  attach to $B_v$ empty-certificates for each view between $\View{T}$ and $v$ \;
  }
\Indm
}

\Versioned{
\caption{(Streamlined) HotStuff-2 Protocol} \label{alg:hs2}
}{
\caption{\Ctail Protocol (differences from HotStuff-2 in blue)} \label{alg:carry}
}
\end{algorithm}

%% file: related.tex
\section{Related Work}

\textbf{Rotational Leader.}
The historical perspective which leads to \Ctail stems from HotStuff~\cite{Yin2019HotStuffBC}, the first atomic broadcast protocol which incurs linear word-communication in steady state, even in the case of a handover from a faulty leader; and a quadratic word communication complexity in the worst case under a cascade of faulty leaders. RareSync~\cite{Civit2023ByzantineCI} and Lewis-Pye~\cite{LewisPye2022QuadraticWM} provide an implementation of HotStuff's pacemaker module (that was left underspecified) that provides view-synchronization with worst-case quadratic communication cost and amortized linear.

The HS-like family of protocols reduces leader replacement communication costs to linear, 
enabling regular leader replacement at no additional communication cost or drop in system 
throughput. 
Additionally, these protocols streamline protocol phases to double the system throughput. 
Variants of HotStuff which maintain linearity include
HotStuff-2~\cite{Malkhi2023HotStuff2OT}, which achieves two-phase latency while maintaining linearity; and HotStuff-1~\cite{Kang2024HotStuff1LC}, which contributes a streamlined variant of HotStuff-2 and speculative fast confirmation. 
Several other protocols have aimed two-phase streamlined and linear latency.
However, Fast-HotStuff~\cite{Jalalzai2020FastHotStuffAF} and Jolteon~\cite{jolteon} have quadratic complexity in view-change; 
AAR~\cite{aar} employs expensive zero-knowledge proofs; Wendy~\cite{Giridharan2021NoCommitPD} relies 
on a new aggregate signature construction (and it is super-linear); Marlin~\cite{marlin} introduces an additional \emph{virtual block}, granting the leader an extra opportunity to propose a block supported by all honest replicas but falls back to 3-phase in the presence of failures.

\textbf{Tail-Forking.}
BeeGees ("BG")~\cite{Giridharan2022BeeGeesSA} indicated that HS-like protocols suffer significant degradation in throughput against the tail-forking problem (the original HotStuff variant may even lose liveness in an extreme scenario). Our work builds on the BG observation and borrows insights largely from it, while aiming to maintain linearity.

BG formulated a property called \textit{``Any-Honest-Leader commit''} (AHL): after GST, once an honest leader proposes in a view, that block will be committed after at most $k$ subsequent honest-leader views\footnote{$k$ is a protocol parameter indicating the number of phases to reach a commit; typically, $k=2$ or $k=3$.}. BG employs a complex leader handover in order to satisfy AHL.   
Rather than AHL, we formulate a property called \textit{$\rho$-tail-resilience}. This property guarantees, under a reasonably small choice of $\rho$ (e.g., $\rho=6$), that a large constant fraction of honest leader proposals become committed even against worst-case tail-forking attacks.
Note that this property is weaker than AHL, but suffices for the protocol to be censorship free and maintain effective good-put, while preserving linearity.
In terms of technical ingredients, 
BG falls back into the regime set forth in PBFT~\cite{Castro1999PracticalBF}. This results in a complex leader handover regime: a view-change incurs quadratic word communication complexity and requires the leader to justify its proposal. A variant of BG replaces the explicit set of 2f+1 message with a SNARK, a complicated procedure with a high computational cost. 
We are somewhat concerned about the of use SNARKs to reduce word-communication complexity for two reasons: (i) throughput is bottle-necked by SNARK generation capacity, which is orders of magnitude lower than traditional BFT consensus throughput, and (ii) the original, uncompressed information is not guaranteed to be available.

The \Carry mechanism aligns with the BG philosophy: use replicas' votes not merely to tally agreement, but to transport knowledge of the last safe block forward.
However, there are two important differences. 

\begin{enumerate}
    \item Instead of leader justification, replicas guard their votes in \Carry. They simply do not vote for a proposal by a future leader unless it extends their highest vote. This follows the same intuition as HotStuff, and thus preserves its two core tenets: linearity and simplicity.
    
    \item The \Carry mechanism allows leaders to reinstate a block even if they don’t have a QC for it. This enhances good throughput against benign slowness, and a proposal by a slow leader can be committed even if it does not receive $2f{+}1$ votes. 
    
\end{enumerate}

{\bf Leader Slowness.}
The leader-slowness attack is a well-known problem in blockchains~\cite{time-role-mev,pbs2024,Daian2019FlashB2}.
Prior work has illustrated that in Ethereum, for $59\%$ of blocks, proposers have earned higher MEV rewards than block rewards~\cite{time-role-mev}, and 
any additional delay in proposing can help maximize their MEVs~\cite{time-is-money}.
There are two popular solutions to tackle leader slowness:
(i) Exclude any block that misses a set deadline to the main blockchain. However, a clever proposer can still delay proposing until the deadline~\cite{fork-based}.
(ii) Assign block rewards proportional to the number of attestations; a delayed block will receive fewer attestations and thus reduced block rewards~\cite{reduced-block-reward}. 
However, if MEV rewards exceed total block rewards, the proposer makes a profit despite losing any block reward.

{\bf Linearity.}
Maintaining linearity of the HotStuff family of solutions is of paramount importance in this paper. Hotstuff linearity has unlocked the first tight upper bounds to the atomic broadcast problem in a variety of settings which were open for decades. 

Concretely, in pure asynchrony, VABA~\cite{Abraham2019AsymptoticallyOV} provides
the first tight upper bound for the Validated Byzantine Agreement problem.
In particular, before HotStuff, the best known communication upper bound for this problem model was due to Cachin et al.~\cite{Cachin2001SecureAE}, 2001, incurring $O(n^3)$ communication complexity in expectation. VABA harnesses HotStuff to arrive at a tight solution whose communication complexity is $O(n^2)$.   

In the Authenticated Channels model (no signatures),  
Information Theoretic HotStuff~\cite{Abraham2020InformationTH} provides the first tight upper bound BA protocol with bounded communication complexity $O((f_a+1) \cdot n^2)$ against $f_a$ actual failures.

{\bf Real-World Deployments.}
In addition to foundational contributions, the defense provided by \Carry against tail-forking may be valuable in practice for real life adoptions of HotStuff. At the time of this writing, blockchain companies which have announced employing some HotStuff variant at their core (to our knowledge) include Diem(Libra), Cypherium, Flow, Celo, Aptos, Espresso Systems, Pocket Network, SpaceComputer, and many others. 
 

%% file: conclusion.tex
\section{Conclusion}
The \Carry mechanism enhances HS-like protocols' robustness by protecting against tail‐forking attacks. It also boosts performance under straggler leaders by ensuring safe progress with aggressive responsiveness without waiting for full quorums.
\Ctail, a full solution that combines these methods with HotStuff-2, exhibits efficient linearity and high performance against both malicious leader and benign transient delays.

%% file: ba-intro.tex
\section{Introduction}

HotStuff-like (HS-like) Byzantine Fault Tolerant (BFT) consensus protocols~\cite{Yin2019HotStuffBC,hotshot,Chan2023SimplexCA,Kang2024HotStuff1LC,Malkhi2023HotStuff2OT} are widely adopted in modern blockchain and decentralized systems for their conceptual simplicity, responsive liveness, linear communication, and censorship resistance. However, frequent leader rotation makes HS-like protocols vulnerable to \emph{tail-forking} attacks~\cite{Giridharan2022BeeGeesSA}, where a malicious or sluggish leader can prevent a previously proposed block from being extended, significantly degrading throughput. Existing defenses either require expensive quadratic communication or computationally heavy SNARKs to prove the absence of quorum certificates~\cite{Jalalzai2025MonadBFTFR,Jalalzai2023VBFTVB}, undermining the hallmark efficiency of HS-like protocols.

We present \textbf{\Carry}, a lightweight, drop-in mechanism for HS-like protocols that defends against tail-forking while preserving linear communication complexity. \Carry elevates votes to first-class citizens, enabling them to transport knowledge of the voted blocks in the last $\rho$ views. This prevents a Byzantine leader from skipping a block without proof, and allows an honest leader to reinstate it without collecting a quorum of votes. We implement \Carry on top of HotStuff-2 to create \Ctail, which achieves provable $\rho$-tail-resilience (see Definition~\ref{def:rho}) and liveness under hostile or sluggish leadership, all with $O(n)$ communication.

\begin{flushleft}\label{def:rho}
\textbf{Definition 1} (\emph{$\rho$-tail-resilience}). \emph{We say that a proposal $T$ is \textit{$\rho$-isolated} if view of $T$ is situated among a succession of $\rho$ consecutive bad leaders.
(After GST,) each proposal $T$ by an honest leader, which is not $\rho$-isolated, is guaranteed to be included in the global sequence.}
\end{flushleft}

%% file: ba-background.tex
\section{Background}
\subsubsection*{Streamlined, Block-based Protocols: Key Concepts}

HS-like protocols combine a \emph{streamlined} design—reducing or overlapping phases, embedding view changes into block proposals, and rotating leaders every view. HS-like protocols adopt a \emph{block-based} structure, where each proposal is a quorum-certified block extending its predecessor. The \emph{streamlined} design improves throughput, latency, simplicity, and achieves linear communication. 

In each view $v$, leader $L_v$ proposes block $B_v$; a \emph{Quorum Certificate (QC)} is formed by aggregating $2f{+}1$ threshold signature shares to prove quorum agreement, and blocks include QCs to ensure safety. The QCs are formed with linear overhead via threshold signatures, and the blocks are chained through the QCs contained.

\subsubsection*{Streamlined HotStuff-2 in a Nutshell}

\begin{figure}
\centering
\includegraphics[width=.7\linewidth]{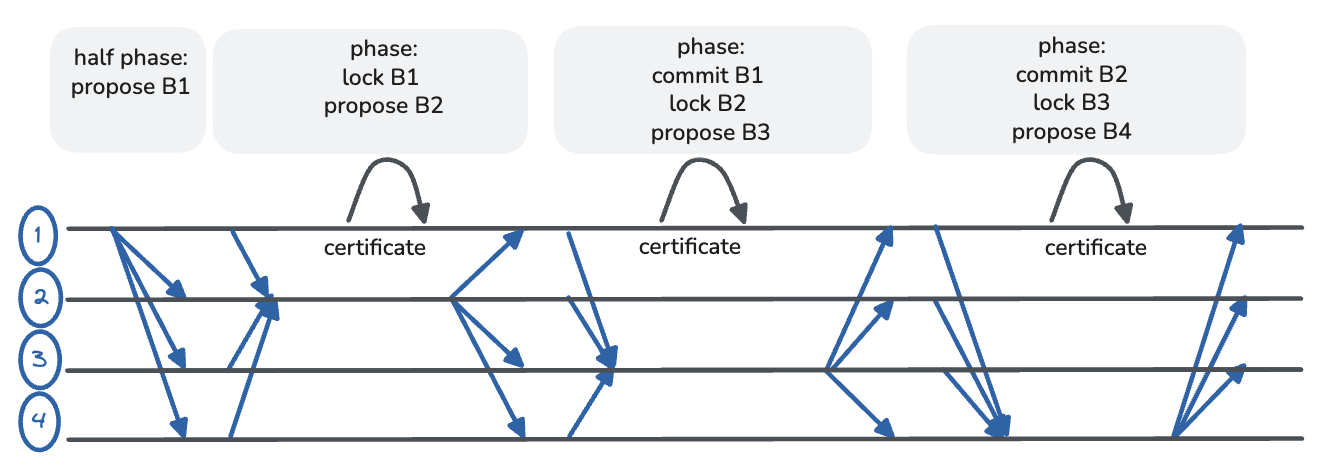}
\caption{Streamlined HotStuff-2 protocol flow.} \label{fig:shs2}
\end{figure} 

We use streamlined HotStuff-2~\cite{Malkhi2023HotStuff2OT}, a variant of HotStuff that reduces latency by two half-phases, to exemplify the \Carry mechanism (Figure~\ref{fig:shs2}), though \Carry applies to other streamlined HS-like protocols. In each view $v$, the leader may propose a block if it either forms a fresh QC by aggregating $2f{+}1$ \texttt{NEW-VIEW} messages with identical votes for some prior view $w<v$, or receives $3f{+}1$ \texttt{NEW-VIEW} messages (or a Pacemaker signal) ensuring messages from all honest replicas. Replicas maintain a \emph{lock}—the highest QC they have voted for—and vote only for blocks containing a QC whose view is no lower than that of their lock, updating their lock on voting. A commit occurs when two consecutive QCs form a chain: $B_v$ is committed if $B_{v+1}$ extends $B_v$ and $B_{v+2}$ extends $B_{v+1}$, via their QCs. The protocol proceeds view-by-view: replicas send the next leader their lock and vote; the leader proposes a block extending the highest QC it aggregated; replicas vote if rules are met, then advance to the next view via the Pacemaker, which coordinates with QCs or timeouts.

\subsection*{The Leader-Slowness Problem and the Tail-Forking Attack}

Two key issues undermine the performance of streamlined protocols: \emph{leader-slowness} and \emph{tail-forking}. Leader-slowness occurs when a leader, due to network delays or limited hardware capacity, is slow to propose blocks, potentially missing its turn; protocols often allow generous timeouts to avoid premature replacement, but this slows responsiveness to actual faults. Tail-forking arises because commits depend on two consecutive leaders: a malicious leader can refuse to extend an honest predecessor’s block, forking the tail and delaying its commitment, which increases latency, reduces throughput, and compromises fairness. These challenges motivate the \Carry mechanism, which mitigates both problems while preserving the streamlined efficiency of HS-like protocols.

%% file: ba-carry.tex
\section{Carry}
\label{sec:carry}

The \Carry mechanism improves the view-change mechanisms by requiring each new leader to \emph{reinstate} the last uncertified block it observes, ensuring progress is not blocked by the leader-induced stalls. In HS-like protocols’ lock-commit regime, a block $B_x$ commits only if $2f{+}1$ replicas lock it by voting for an extending block $B$ (e.g., $B.qc = \QC{B_x}$ in HotStuff-2). While these votes are sent in \texttt{NEW-VIEW} messages, current rules do not oblige the next leader to honor them, allowing sluggish or Byzantine leaders to ignore votes and cause tail-forking. \Carry treats votes as first-class citizens: replicas guard their highest vote and carry it into the next view, and the leader must extend it or reinstate it as $B_v.reinstated = T$, protecting both the certified block and its extension from being silently dropped.

\subsection*{\Carry Implementation}
\mypp{Carry Rule.}
A valid proposal $B_v$ from the leader $L_v$ of view $v$ has the following format:
\begin{description}
    \item[Highest QC:] the highest quorum certificate $\QC{x}$ known to the leader $L_v$ is attached as $B_v.qc = \QC{x}$; 
    \item[Reinstate:] if $L_v$ received fewer than $2f{+}1$ votes for the highest tail $T$ extending $B_x$, and $v - x \leq \rho$, then $T$ is reinstated (in full) as $B_v.reinstated$;
    \item[Justification:] If $v - x \leq \rho$, then empty-certificates are attached to $B_v$ for each view, which $B_v$ does \textbf{not} extend, between views $x$ and $v$. It is worth noting that if $B_v.reinstated$ exists, empty-certificates for views preceding the reinstated block are recursively attached within the chain of reinstated blocks.
\end{description}

\mypp{Voting-Rule}
       
A replica accepts the proposal $B_v$ of leader $L_v$ of view $v$ if:
\begin{enumerate}
    \item $B_v.qc$ has a higher or equal view than the replica's $lock$,
    \item $B_v$ adheres to the Carry Rule above.
\end{enumerate}

\Carry retains from HotStuff-2 the Commit Rule and Leader Proposal Rule.

\subsection{The \Ctail Protocol}

Akin to HotStuff-2, the \Ctail protocol flows view-by-view. 
At the end of each view, the replicas and the incoming leader perform a handover protocol as follows:

\begin{description}

\item[\textbf{Replica $\longrightarrow$ incoming leader.}]

Each replica sends an incoming leader a \texttt{NEW‑VIEW} message that carries
\begin{enumerate}
    \item the next view number
    \item the replica’s highest QC (its \emph{lock})
    \item its vote-shares (possibly empty) in the past $\rho$ views.\footnote{It is possible to send less information if the lock precedes the current view by less than $\rho$ views; we omit this optimization for simplicity.}
\end{enumerate}
      
\item[\textbf{Leader $\longrightarrow$ replicas.}]

On satisfying the Leader Proposal Rule, the leader of view $v$ proposes a block $B_v$ that   
\begin{enumerate}
    \item extends the highest QC it has collected, potentially freshly aggregated from \texttt{NEW-VIEW} messages,  as $B_v.qc$,
    \item reinstates $T$ in full as $B_v.reinstated$, provided that the highest QC is from the past $\rho$ views and the highest voted block $T$ extends the highest QC, 
    \item attaches empty-certificates for each view between $\View{T}$ and $v$.

\end{enumerate}

\end{description}

Note that if a leader receives two conflicting highest votes, it reinstates
\emph{both}; honest replicas will recognize the view as faulty and drop the conflicting votes.

\subsection*{Correctness Proofs}

We envision \Carry as a drop-in mechanism for any HS-like protocol, and in this work use HotStuff-2 to illustrate its design. Since \Carry only enhances performance without altering fundamental safety and liveness rules, \Ctail inherits these guarantees from HotStuff-2. We defer detailed arguments to Appendix~\ref{s:appendix} and focus here on proving the key novel property of \Ctail: $\rho$-tail-resilience.

\begin{theorem}
    After GST, a tail block $T$ proposed by an honest leader that is not $\rho$-isolated and receives $2f{+}1$ votes will be extended by the next honest view. 
\end{theorem}

\begin{proof}
    \Dakai{The proof of the Other Case is a bit confusing for me? I would write a new version soon.}
    As $T$ is not $\rho$-isolated, there exists two honest views $x$ and $y$, such that $x < \View{T} < y \leq x + \rho$. We want to show that the honest leader $L_y$'s proposal $B_y$ extends $T$ in the chain, $B_y$.qc $= \QC{T}$.  
    %

    {\em Easy Case:} When $y = \View{T}{+}1$, $L_y$ forms a QC for $T$ as the highest QC, proposes a block $B_y$ that extends $\QC{T}$ and disseminates it. All honest replicas become locked on $B_y.qc$, and $T$ is guarded from ever being skipped. 

    {\em Other Case:} When $y > \View{T}+1$ and $L_y$ obtains a QC for a view higher than $\View{T}$. 
    We need to show that $B_y$ does not skip $T$, i.e., that $B_y.qc \succ T$. 
    We do so by showing that between views $x$ and $y$ if a proposal $P$ extends $T.qc$ and receives $f{+}1$ honest votes, then $P$ is extended by the next valid proposal $P'$. {\em Note:} since $L_x$ is honest and $T.qc$ is from a view $x$ or higher, the number of views between $T.qc$ and $P'$ do not exceed $\rho$.

    By the Carry Rule, if $P'.reinstated$ exists, $P'$ must include an empty-certificate as a justification for each view between $P'.reinstated$ and $P'$. Because $P'$ is the next valid proposal succeeding $P$, $P'.reinstated$ has to be from a view not higher than $P$. But since $P$ has $f{+}1$ honest votes, it is impossible for $P'$ to have an empty certificate for it. Hence, $P'.reinstated$ must be $P$. 
If $P'.reinstated$ does exist, then $P'$ must have a fresh QC. Again, by minimality, $P'.qc$ must not be higher than $P$, but cannot skip $P$ either. Hence, in this case, $P'.qc$ is for $P$.

From the argument above, there may be a succession of valid proposals between $x$ and $y$ extending $T.qc$ that are chained to each other via QCs or reinstated blocks. $T$ is within this chain, and $L_y$ terminates it. We obtain that $L_y \succ $T.
\end{proof}

\mypp{Illustration.}

Consider the case of an honest tail $T$ referencing a QC from a view $x < \View{T}$. We want to prevent a bad leader in a future view $z$ from tail-forking $T$. 
We assume that views $x$ and $z$ are not more than $\rho$ views apart, i.e., $z \leq x + \rho$.
Figure~\ref{fig:carry} illustrates diagrammatically 
how the \Carry mechanism prevents tail-forking of $T$ under three possible cases:

\begin{figure}[t]
    \centering
    \includegraphics[width=0.65\linewidth]{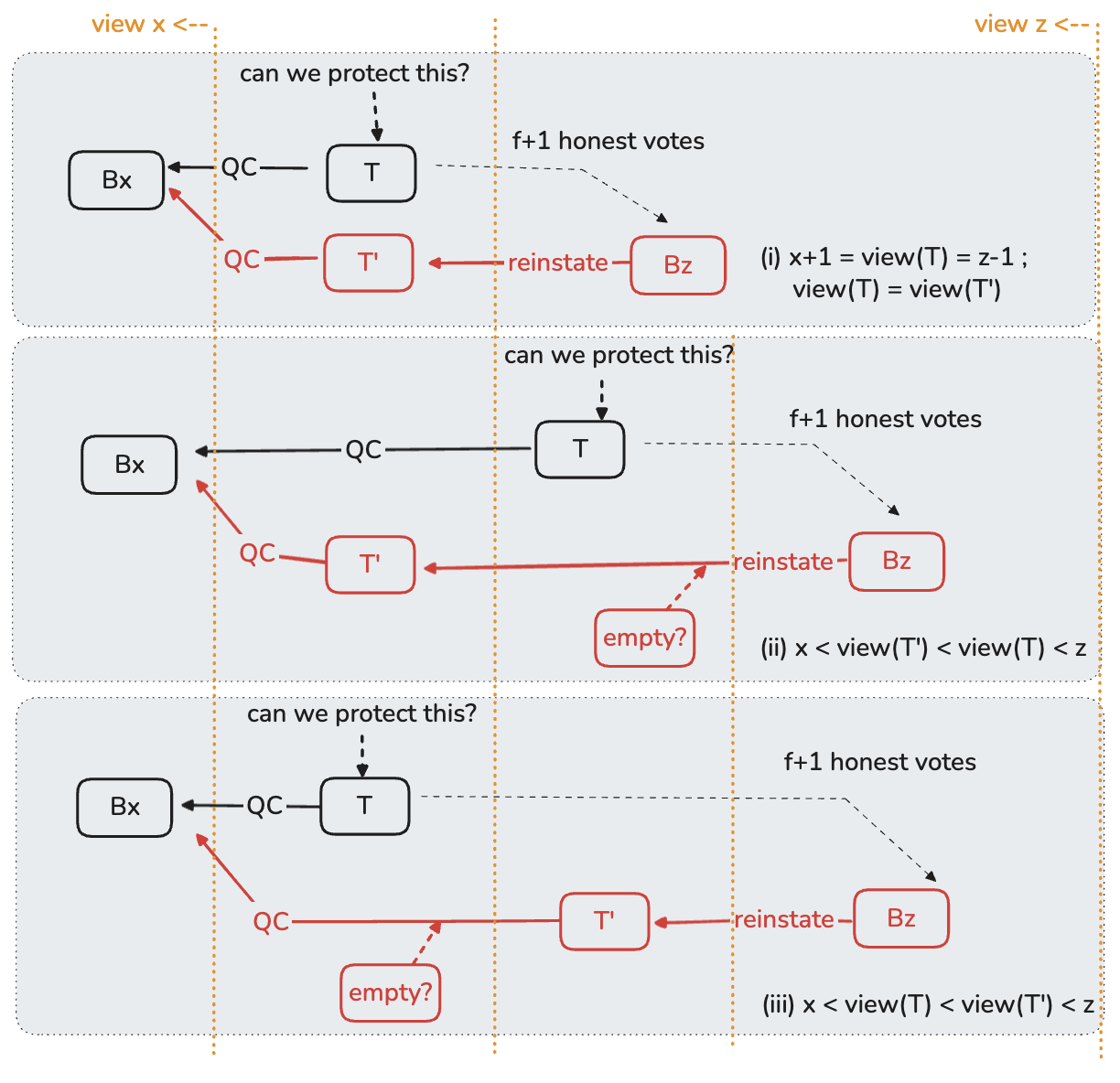}
    \caption{Attempts by $L_z$ to hide $T$ by attaching an invalid reinstated-block $T'$ which skips $T$.}
    \label{fig:carry}
\end{figure}

\begin{description}
  \item[(i) Consecutive views.]
        If $z = x + 2$, then $\mathrm{view}(T') = \mathrm{view}(T)$ and forking is impossible.
  \item[(ii) Skipping forward.]
        If $\mathrm{view}(T') < \mathrm{view}(T)$, the $f{+}1$ honest replicas who voted for $T$
        will refuse to vote for $B_z$ due to the lack of empty-certificate for $\mathrm{view}(T)$.
  \item[(iii) Skipping backward.]
        If $\mathrm{view}(T') > \mathrm{view}(T)$, those same $f{+}1$ honest replicas will not sign an
        empty-certificate for view $\mathrm{view}(T)$; $L_z$ thus lacks the proof
        mandated in the Carry Rule: it will miss an empty-certificate for $\mathrm{view}(T)$ which is between the locked-view $x$ and the alleged highest reinstated $\mathrm{view}(T')$. 
\end{description}

Thus, any successful tail-forking attack requires \emph{$\rho{+}1$ consecutive Byzantine leaders}. Otherwise, \Carry forces each proposal to acknowledge the highest vote or certificate, blocking faulty leaders from bypassing honest blocks. Moreover, \Carry alleviates leader-slowness: even if an honest leader $L_v$ times out before reaching $2f{+}1$ replicas, its block $B_v$ cannot be skipped as long as at least $f{+}1$ honest votes exist and $L_v$ is not \emph{$\rho$-isolated}.

%% file: ba-related.tex
\section{Related Work}

\textbf{Rotational Leader and Linearity.}
HotStuff~\cite{Yin2019HotStuffBC} pioneered linear word communication in steady state, even across faulty leader handovers. RareSync~\cite{Civit2023ByzantineCI} and Lewis-Pye~\cite{LewisPye2022QuadraticWM} later refined HotStuff’s pacemaker with quadratic worst-case but amortized linear costs. HS-like protocols improve this by reducing leader replacement to linear cost and streamlining phases for higher throughput. Variants such as HotStuff-2~\cite{Malkhi2023HotStuff2OT} and HotStuff-1~\cite{Kang2024HotStuff1LC} achieve two-phase latency or speculative fast confirmation while retaining linearity. Other approaches (e.g., Fast-HotStuff~\cite{Jalalzai2020FastHotStuffAF}, Jolteon~\cite{jolteon}, Wendy~\cite{Giridharan2021NoCommitPD}, Marlin~\cite{marlin}, AAR~\cite{aar}) trade simplicity or linearity for two-phase latency, often falling back to quadratic view-change or expensive zero-knowledge proofs. Maintaining linearity is crucial, as HotStuff has enabled the first tight communication bounds for atomic broadcast in both asynchronous and authenticated-channel models~\cite{Abraham2019AsymptoticallyOV,Abraham2020InformationTH}.

\textbf{Tail-Forking and Leader Slowness.}
\textbf{Tail-Forking.}
BeeGees ("BG")~\cite{Giridharan2022BeeGeesSA} indicated that HS-like protocols suffer significant degradation in throughput against the tail-forking problem (the original HotStuff variant may even lose liveness in an extreme scenario). Our work builds on the BG observation and borrows insights largely from it, while aiming to maintain linearity.

BG formulated a property called \textit{``Any-Honest-Leader commit''} (AHL): after GST, once an honest leader proposes in a view, that block will be committed after at most $k$ subsequent honest-leader views\footnote{$k$ is a protocol parameter indicating the number of phases to reach a commit; typically, $k=2$ or $k=3$.}. BG employs a complex leader handover in order to satisfy AHL.   
Rather than AHL, we formulate a property called \textit{$\rho$-tail-resilience}. This property guarantees, under a reasonably small choice of $\rho$ (e.g., $\rho=6$), that a large constant fraction of honest leader proposals become committed even against worst-case tail-forking attacks.
Note that this property is weaker than AHL, but suffices for the protocol to be censorship free and maintain effective good-put, while preserving linearity.

In terms of technical ingredients, 
BG falls back into a complex leader handover regime: a view-change incurs quadratic word communication complexity and requires the leader to justify its proposal. A variant of BG replaces the explicit set of $2f+1$ messages with a SNARK, a complicated procedure with high computational cost. 

The \Carry mechanism aligns with the BG philosophy: use replicas' votes not merely to tally agreement, but to transport knowledge of the last safe block forward.
However, there are two important differences. 

\begin{enumerate}
    \item Instead of leader justification, replicas guard their votes in \Carry. They simply do not vote for a proposal by a future leader unless it extends their highest vote. This follows the same intuition as HotStuff, and thus preserves its two core tenets: linearity and simplicity.
    
    \item The \Carry mechanism allows leaders to reinstate a block even if they don’t have a QC for it. This enhances good throughput against benign slowness, and a proposal by a slow leader can be committed even if it does not receive $2f{+}1$ votes. 
\end{enumerate}

{\bf Leader Slowness.}
The leader-slowness problem is well known in blockchains~\cite{time-role-mev,pbs2024,Daian2019FlashB2}. In Ethereum, proposers earned higher MEV than block rewards in 59\% of blocks~\cite{time-role-mev}, and delaying proposals can further boost MEV~\cite{time-is-money}. Two common countermeasures are: (i) discarding blocks missing a deadline~\cite{fork-based}, though proposers can still wait until the deadline; and (ii) rewarding blocks proportionally to attestations~\cite{reduced-block-reward}, which penalizes late blocks. Yet, if MEV exceeds block rewards, proposers still profit despite these penalties.

%% file: ba-conclusion.tex
\section{Conclusion}
The \Carry mechanism enhances HS-like protocols' robustness by protecting against tail‐forking attacks. It also boosts performance under straggler leaders by ensuring safe progress with aggressive responsiveness without waiting for full quorums.

%% file: proofs.tex
\section{Correctness Proofs (continued)}
\label{s:appendix}

First, we discuss the guarantees for non-equivocation and non-conflicting commitment. 

\begin{lemma}\label{lem:no-equivocation}
    Let $R_1$ and $R_2$ be two honest replicas that lock certificates $\QC{v}^1$ and $\QC{v}^2$ of view $v$, respectively. 
    \Ctail guarantees that $\QC{v}^1 = \QC{v}^2$.
\end{lemma}

\begin{proof}
    An honest replica $R_i$ locks a certificate $\QC{v}^i$ if the Voting Rule is true, which can only happen 
    if $\QC{v}^i$ is attached to the proposal sent by the leader for view $v{+}1$ (or higher). 
    Each of these QC's are composed of threshold signature-shares of $2f{+}1$ replicas. 
    
    Let $S_i$ be the replicas that voted for certificate $\QC{v}^i$.
    Let $X_i = S_i - f$ be the honest replicas in $S_i$. 
    As $|S_i| = 2f{+}1$, we have $|X_i| = 2f{+}1 - f$.
    If $\QC{v}^1 \neq \QC{v}^2$, then $X_1$ and $X_2$ must not overlap. 
    Hence, $|X_1 \cup X_2| \geq 2(2f{+}1 - f)$. 
    This simplifies to $|X_1 \cup X_2| \geq 2f{+}2$, which contradicts $n = 3f{+}1$.
    Hence, we conclude $\QC{v}^1 = \QC{v}^2$.
\end{proof}

\begin{lemma}\label{lem:no-conflict-cert}
    If a replica $R$ receives a certificate $\QC{v+1}$ that extends certificate $\QC{v}$,
    then no certificate $\QC{w}$ conflicts with $\QC{v}$, where view $w > v$, can exist.
\end{lemma}

\begin{proof}
    If a replica $R$ received $\QC{v+1}$ that extends $\QC{v}$, 
    it implies that $2f{+}1$ replicas that set $\QC{v}$ as their highest lock agreed to vote for $\QC{v+1}$.
    Let's denote the honest replicas in these $2f{+}1$ replicas as $A$, who will not vote for a block conflicting with $B_v$.
    
    We assume that $\QC{w}$ is the lowest QC that conflicts with $\QC{v}$ such that $w > v+1$. For $\QC{w}$ to exist, there must be $2f{+}1$ replicas who voted for it.
    Let's denote the honest replicas from these $2f{+}1$ replicas as $A'$.

    As $A$ will not vote for $B_w$, thereby $A$ and $A'$ do not overlap. Then,
    $|A \cup A'| \geq 2(2f{+}1 - f)$. This simplifies to $|A \cup A'| \geq 2f{+}2$, which contradicts $n = 3f{+}1$.
    Hence, we conclude that $\QC{w}$ could not exist.
\end{proof}

\begin{corollary}\label{cor:unique-commit}
    If an honest replica $R$ commits a block $B_v$, then no other conflicting block can commit.
\end{corollary}
\begin{proof}
    Assume a block $B_{w}$, proposed in view $w$, conflicts with block $B_v$ and another honest replica $R'$ commits $B_w$.
    This implies that replicas $R$ and $R'$ have conflicting states.
    To commit $B_v$ and $B_w$, $R$ and $R'$ must follow the Commit Rule, respectively: $R$ must receive two consecutive QCs $\QC{y+1}$ extending $\QC{y}$ such that $y\ge x$, and $R'$ must receive two consecutive QCs $\QC{z+1}$ extending $\QC{z}$ such that $z\ge w$. As $B_v$ conflicts with $B_w$, then obviously $\QC{y}$ conflicts with $\QC{z}$. If we assume that $y > z+1$, 
    from Lemma~\ref{lem:no-conflict-cert}, we know that if $\QC{z}$ and $\QC{z+1}$ exist, then $\QC{y}$ cannot exist, which contradicts the fact that $R$ commits $B_v$; If we assume that $z > y+1$, similarly, it contradicts the assumption that $R'$ commits $B_w$.
\end{proof}

Using these lemmas, next, we argue \Ctail's safety guarantees, which in blockchain systems ensure that at each entry in the global ledger (or log) there is a unique block.

\begin{theorem}
    \Ctail guarantees consensus safety in a system with $n \ge 3f{+}1$ replicas:
     if two honest replicas $R_1$ and $R_2$ commit blocks $A$ and $B$, respectively, at the same position $k$ in the ledger, then $A$ = $B$.
\end{theorem}

\begin{proof}
    Lemma~\ref{lem:no-equivocation} helps to illustrate that within a view a leader cannot equivocate, that is, 
    only one block can commit in a view. 
    Furthermore, from Corollary~\ref{cor:unique-commit}, we know that two conflicting committed blocks cannot exist.
    This implies that both $A$ and $B$ are permanently part of the respective global ledgers of $R_1$ and $R_2$ and must not conflict.

    Now, assume that $A \neq B$ and $B$ extends $A$. 
    Moreover, we know that $A$ is at the position $k$ in the ledger.
    Therefore, in the ledger, $B$ should succeed $A$. 
    As $B$ is also at position $k$ in the ledger, it implies that at least $k{-}1$ blocks precede $B$ (including $A$), 
    while at most $k{-}2$ blocks precede $A$. 
    However, this contradicts the assumption that $A$ is committed at the position $k$ in the ledger with $k-1$ blocks preceding it.
\end{proof}

Next, we discuss the liveness guarantee of the \Ctail protocol.

\begin{lemma}~\label{lemma:learnall}
    After GST, an honest leader of view $v$ can learn the highest QC across all honest replicas, all votes, and empty certificates for the past $\rho$ views before $v$.
\end{lemma}

\begin{proof}
    The View Synchronization mechanisms~\cite{raresync,LewisPye2022QuadraticWM} ensure that after GST the leader can receive \texttt{NEW-VIEW} messages from at least $2f{+}1$ honest replicas. 
    
    As each \texttt{NEW-VIEW} message from a replica $R$ contains the \emph{highest-QC} and \emph{highest-vote} known to $R$ and vote-shares of past $\rho$ views before $v$. Thus, for past $\rho$ views between the \emph{highest-QC} and the current view $v$, the leader can either form empty certificates or learn votes. Also, the leader can learn the highest QC across all honest replicas or aggregate a higher one using collcted votes.
\end{proof}

\begin{lemma}~\label{lemma:voteall}
    After GST, the $B_v$ from an honest leader of view $v$ will receive votes from all honest replicas.
\end{lemma}

\begin{proof}
    The View Synchronization mechanisms ensure that after GST, each proposal arrives at all honest replicas before timeout. 
   From Lemma~\ref{lemma:learnall}, we know that each proposal contains the highest QC across all honest replicas, meeting the first part of the Voting Rule. 

   The second part checks the Carry Rule: if the QC is from a view that is more than $\rho$ views earlier than $v$, then the Carry Rule is met. 
   Otherwise, according to the protocol and Lemma~\ref{lemma:learnall}, the proposal from an honest leader would contain (i) a reinstated block (highest-vote) extending the highest-QC, and (ii) all empty certificates between the reinstated block and view $v$, meeting the Carry Rule.

   Thus, all honest replicas will satisfy the Voting Rule and will vote for the proposal.
\end{proof}

\begin{theorem}
    After GST, there is an unbounded number of committed blocks proposed by honest leaders.
\end{theorem}

\begin{proof}
Following Lemma~\ref{lemma:voteall}, we can conclude that any sequence of three consecutive honest leaders $L_v$, $L_{v+1}$, and $L_{v+2}$ will generate two consecutive QCs: $\QC{v}$ and $\QC{v+1}$, thereby committing block $B_v$ and its prefix.

Given that the system comprises $n = 3f + 1$ replicas, such a trio of consecutive honest leaders is guaranteed to occur at least once in every sequence of $n$ views. As the protocol proceeds, this ensures an unbounded number of block commitments.
\end{proof}

%% file: main.bbl
\begin{thebibliography}{10}

\bibitem{Abraham2019AsymptoticallyOV}
Ittai Abraham, Dahlia Malkhi, and Alexander Spiegelman.
\newblock Asymptotically optimal validated asynchronous byzantine agreement.
\newblock {\em Proceedings of the 2019 ACM Symposium on Principles of
  Distributed Computing}, 2019.
\newblock URL: \url{https://api.semanticscholar.org/CorpusID:197660727}.

\bibitem{Abraham2020InformationTH}
Ittai Abraham and Gilad Stern.
\newblock Information theoretic hotstuff.
\newblock In {\em International Conference on Principles of Distributed
  Systems}, 2020.
\newblock URL: \url{https://api.semanticscholar.org/CorpusID:221970651}.

\bibitem{aar}
Mark Abspoel, Thomas Attema, and Matthieu Rambaud.
\newblock Malicious security comes for free in consensus with leaders.
\newblock {\em Cryptology ePrint Archive}, 2020.

\bibitem{fork-based}
Aditya Asgaonkar.
\newblock {Proposer LMD Score Boosting, Ethereum Consensus-Specs.}, 2021.
\newblock URL: \url{https://github.com/ethereum/consensus-specs/pull/2730}.

\bibitem{hotshot}
Jeb Bearer, Benedikt B{\"u}nz, Philippe Camacho, Binyi Chen, Ellie Davidson,
  Ben Fisch, Brendon Fish, Gus Gutoski, Fernando Krell, Chengyu Lin, et~al.
\newblock The espresso sequencing network: Hotshot consensus, tiramisu
  data-availability, and builder-exchange.
\newblock {\em Cryptology ePrint Archive}, 2024.

\bibitem{Cachin2001SecureAE}
Christian Cachin, Klaus Kursawe, Frank Petzold, and Victor Shoup.
\newblock Secure and efficient asynchronous broadcast protocols.
\newblock In {\em Annual International Cryptology Conference}, 2001.
\newblock URL: \url{https://api.semanticscholar.org/CorpusID:18716687}.

\bibitem{Castro1999PracticalBF}
Miguel Castro.
\newblock Practical byzantine fault tolerance.
\newblock In {\em USENIX Symposium on Operating Systems Design and
  Implementation}, 1999.
\newblock URL: \url{https://api.semanticscholar.org/CorpusID:221599614}.

\bibitem{pbftj}
Miguel Castro and Barbara Liskov.
\newblock Practical {B}yzantine fault tolerance and proactive recovery.
\newblock {\em ACM Trans. Comput. Syst.}, 20(4):398--461, 2002.
\newblock \href {https://doi.org/10.1145/571637.571640}
  {\path{doi:10.1145/571637.571640}}.

\bibitem{Chan2023SimplexCA}
Benjamin~Y. Chan and Rafael Pass.
\newblock Simplex consensus: A simple and fast consensus protocol.
\newblock {\em IACR Cryptol. ePrint Arch.}, 2023:463, 2023.
\newblock URL: \url{https://api.semanticscholar.org/CorpusID:259092405}.

\bibitem{Chan2020StreamletTS}
Benjamin~Y. Chan and Elaine Shi.
\newblock Streamlet: Textbook streamlined blockchains.
\newblock {\em Proceedings of the 2nd ACM Conference on Advances in Financial
  Technologies}, 2020.
\newblock URL: \url{https://api.semanticscholar.org/CorpusID:211478313}.

\bibitem{raresync}
Pierre Civit, Muhammad~Ayaz Dzulfikar, Seth Gilbert, Vincent Gramoli, Rachid
  Guerraoui, Jovan Komatovic, and Manuel Vidigueira.
\newblock Byzantine consensus is $\theta(n^2)$: The dolev-reischuk bound is
  tight even in partial synchrony!
\newblock In {\em 36th International Symposium on Distributed Computing (DISC
  2022)}, volume 246 of {\em Leibniz International Proceedings in Informatics
  (LIPIcs)}, pages 14:1--14:21. Schloss Dagstuhl, 2022.
\newblock \href {https://doi.org/10.4230/LIPIcs.DISC.2022.14}
  {\path{doi:10.4230/LIPIcs.DISC.2022.14}}.

\bibitem{Civit2023ByzantineCI}
Pierre Civit, Muhammad~Ayaz Dzulfikar, Seth Gilbert, Vincent Gramoli, Rachid
  Guerraoui, Jovan Komatovic, and Manuel Vidigueira.
\newblock Byzantine consensus is $\theta$ (n2): the dolev-reischuk bound is
  tight even in partial synchrony!
\newblock {\em Distributed Comput.}, 37:89--119, 2023.
\newblock URL: \url{https://api.semanticscholar.org/CorpusID:266258469}.

\bibitem{Daian2019FlashB2}
Philip Daian, Steven Goldfeder, Tyler Kell, Yunqi Li, Xueyuan Zhao, Iddo
  Bentov, Lorenz Breidenbach, and Ari Juels.
\newblock Flash boys 2.0: Frontrunning, transaction reordering, and consensus
  instability in decentralized exchanges.
\newblock {\em ArXiv}, abs/1904.05234, 2019.
\newblock URL: \url{https://api.semanticscholar.org/CorpusID:121212213}.

\bibitem{diembft-hotstuff}
Diem.
\newblock {DiemBFT consensus protocol}, 2020.
\newblock URL: \url{https://github.com/diem/diem/tree/latest/consensus}.

\bibitem{jolteon}
Rati Gelashvili, Lefteris Kokoris-Kogias, Alberto Sonnino, Alexander
  Spiegelman, and Zhuolun Xiang.
\newblock Jolteon and {D}itto: Network-adaptive efficient consensus with
  asynchronous fallback.
\newblock In {\em International conference on financial cryptography and data
  security}, pages 296--315. Springer, 2022.

\bibitem{Giridharan2021NoCommitPD}
Neil Giridharan, Heidi Howard, Ittai Abraham, Natacha Crooks, and Alin Tomescu.
\newblock No-commit proofs: Defeating livelock in bft.
\newblock {\em IACR Cryptol. ePrint Arch.}, 2021:1308, 2021.
\newblock URL: \url{https://api.semanticscholar.org/CorpusID:238479346}.

\bibitem{Giridharan2022BeeGeesSA}
Neil Giridharan, Florian Suri-Payer, Matthew Ding, Heidi Howard, Ittai Abraham,
  and Natacha Crooks.
\newblock Beegees: Stayin' alive in chained bft.
\newblock {\em Proceedings of the 2023 ACM Symposium on Principles of
  Distributed Computing}, 2022.
\newblock URL: \url{https://api.semanticscholar.org/CorpusID:256274482}.

\bibitem{Jalalzai2025MonadBFTFR}
Mohammad~Mussadiq Jalalzai and Kushal Babel.
\newblock Monadbft: Fast, responsive, fork-resistant streamlined consensus.
\newblock {\em ArXiv}, abs/2502.20692, 2025.
\newblock URL: \url{https://api.semanticscholar.org/CorpusID:276724832}.

\bibitem{Jalalzai2023VBFTVB}
Mohammad~Mussadiq Jalalzai, Chen Feng, and Victoria Lemieux.
\newblock Vbft: Veloce byzantine fault tolerant consensus for blockchains.
\newblock {\em ArXiv}, abs/2310.09663, 2023.
\newblock URL: \url{https://api.semanticscholar.org/CorpusID:264146016}.

\bibitem{Jalalzai2020FastHotStuffAF}
Mohammad~Mussadiq Jalalzai, Jianyu Niu, Chen Feng, and Fangyu Gai.
\newblock Fast-hotstuff: A fast and robust bft protocol for blockchains.
\newblock {\em IEEE Transactions on Dependable and Secure Computing},
  21:2478--2493, 2020.
\newblock URL: \url{https://api.semanticscholar.org/CorpusID:238260298}.

\bibitem{Kang2024HotStuff1LC}
Dakai Kang, Suyash Gupta, Dahlia Malkhi, and Mohammad Sadoghi.
\newblock Hotstuff-1: Linear consensus with one-phase speculation.
\newblock {\em ArXiv}, abs/2408.04728, 2024.
\newblock URL: \url{https://api.semanticscholar.org/CorpusID:271843554}.

\bibitem{LewisPye2022QuadraticWM}
Andrew Lewis-Pye.
\newblock Quadratic worst-case message complexity for state machine replication
  in the partial synchrony model.
\newblock {\em ArXiv}, abs/2201.01107, 2022.
\newblock URL: \url{https://api.semanticscholar.org/CorpusID:245668696}.

\bibitem{Malkhi2023HotStuff2OT}
Dahlia Malkhi and Kartik Nayak.
\newblock Hotstuff-2: Optimal two-phase responsive bft.
\newblock 2023.
\newblock URL: \url{https://api.semanticscholar.org/CorpusID:259144145}.

\bibitem{time-role-mev}
Burak \"{O}z, Benjamin Kraner, Nicol\`{o} Vallarano, Bingle~Stegmann Kruger,
  Florian Matthes, and Claudio~Juan Tessone.
\newblock Time moves faster when there is nothing you anticipate: The role of
  time in mev rewards.
\newblock In {\em Proceedings of the 2023 Workshop on Decentralized Finance and
  Security}, DeFi '23, page 1–8, New York, NY, USA, 2023. Association for
  Computing Machinery.
\newblock \href {https://doi.org/10.1145/3605768.3623563}
  {\path{doi:10.1145/3605768.3623563}}.

\bibitem{Reiter1994SecureAP}
Michael~K. Reiter.
\newblock Secure agreement protocols: reliable and atomic group multicast in
  rampart.
\newblock In {\em Conference on Computer and Communications Security}, 1994.
\newblock URL: \url{https://api.semanticscholar.org/CorpusID:1990309}.

\bibitem{pbs2024}
Ethereum Roadmap.
\newblock Proposer-builder separation, 2024.
\newblock URL: \url{https://ethereum.org/en/roadmap/pbs/}.

\bibitem{reduced-block-reward}
Caspar Schwarz-Schilling.
\newblock {Retroactive Proposer Rewards}, 2022.
\newblock URL: \url{https://notes.ethereum.org/@casparschwa/S1vcyXZL9}.

\bibitem{time-is-money}
Caspar Schwarz-Schilling, Fahad Saleh, Thomas Thiery, Jennifer Pan, Nihar Shah,
  and Barnab\'{e} Monnot.
\newblock {Time Is Money: Strategic Timing Games in Proof-Of-Stake Protocols}.
\newblock In {\em 5th Conference on Advances in Financial Technologies (AFT
  2023)}, volume 282 of {\em Leibniz International Proceedings in Informatics
  (LIPIcs)}, pages 30:1--30:17, Dagstuhl, Germany, 2023. Schloss Dagstuhl --
  Leibniz-Zentrum f{\"u}r Informatik.
\newblock \href {https://doi.org/10.4230/LIPIcs.AFT.2023.30}
  {\path{doi:10.4230/LIPIcs.AFT.2023.30}}.

\bibitem{marlin}
Xiao Sui, Sisi Duan, and Haibin Zhang.
\newblock Marlin: Two-phase {BFT} with linearity.
\newblock In {\em 2022 52nd Annual IEEE/IFIP International Conference on
  Dependable Systems and Networks (DSN)}, pages 54--66, 2022.
\newblock \href {https://doi.org/10.1109/DSN53405.2022.00018}
  {\path{doi:10.1109/DSN53405.2022.00018}}.

\bibitem{Yin2019HotStuffBC}
Maofan Yin, Dahlia Malkhi, Michael~K. Reiter, Guy Golan-Gueta, and Ittai
  Abraham.
\newblock Hotstuff: Bft consensus with linearity and responsiveness.
\newblock {\em Proceedings of the 2019 ACM Symposium on Principles of
  Distributed Computing}, 2019.
\newblock URL: \url{https://api.semanticscholar.org/CorpusID:197644531}.

\end{thebibliography}
